\documentclass[submission,copyright,creativecommons]{eptcs}
\providecommand{\event}{CL\&C2016} 

\usepackage{hyperref}
\usepackage{doi}
\usepackage{amssymb,amsmath}
\usepackage{amsthm}
\usepackage{amsrefs}
\usepackage{enumerate}

\newtheorem{thm}{Theorem}

\newtheorem{cor}[thm]{Corollary}
\newtheorem{defi}[thm]{Definition}
\newtheorem{rem}[thm]{Remark}
\newtheorem{nota}[thm]{Notation}

\newtheorem{tempie}[thm]{Template}
\newtheorem{ack}[thm]{Acknowledgement}

\newcommand\be{\begin{equation}}
\newcommand\ee{\end{equation}}

\newbox\gnBoxA
\newdimen\gnCornerHgt
\setbox\gnBoxA=\hbox{$\ulcorner$}
\global\gnCornerHgt=\ht\gnBoxA
\newdimen\gnArgHgt

\def\Godelnum #1{%
	\setbox\gnBoxA=\hbox{$#1$}%
	\gnArgHgt=\ht\gnBoxA%
	\ifnum \gnArgHgt<\gnCornerHgt
		\gnArgHgt=0pt%
	\else
		\advance \gnArgHgt by -\gnCornerHgt%
	\fi
	\raise\gnArgHgt\hbox{$\ulcorner$} \box\gnBoxA %
		\raise\gnArgHgt\hbox{$\urcorner$}}
\def\bdefi{\begin{defi}\rm}
\def\edefi{\end{defi}}
\def\bnota{\begin{nota}\rm}
\def\enota{\end{nota}}
\def\brem{\begin{rem}\rm}
\def\erem{\end{rem}}

\def\FIVE{\Pi_{1}^{1}\text{-\textsf{CA}}_{0}}

\def\intern{\textup{\textsf{int}}}
\def\IST{\textup{\textsf{IST}}}
\def\PCM{\textup{\textsf{PCM}}}
\def\ZFC{\textup{\textsf{ZFC}}}
\def\ATR{\textup{\textsf{ATR}}}

\def\STP{\textup{\textsf{STP}}}

\def\DNR{\textup{\textsf{DNR}}}

\def\H{\textup{\textsf{H}}}
\def\RCA{\textup{\textsf{RCA}}}

\def\ef{\textup{\textsf{ef}}}
\def\ns{\textup{\textsf{ns}}}

\def\WKL{\textup{\textsf{WKL}}}

\def\bye{\end{document}}

\def\P{\textup{\textsf{P}}}

\def\R{{\mathbb  R}}

\def\MUC{\textup{\textsf{MUC}}}

\def\MCT{\textup{\textsf{MCT}}}

\def\R{{\mathbb{R}}}
\def\({\textup{(}}
\def\){\textup{)}}

\def\st{\textup{st}}

\def\asa{\leftrightarrow}

\def\di{\rightarrow}

\def\ACA{\textup{\textsf{ACA}}}
\def\paai{\Pi_{1}^{0}\textup{-\textsf{TRANS}}}
\def\Paai{\Pi_{1}^{1}\textup{-\textsf{TRANS}}}

\def\QFAC{\textup{\textsf{QF-AC}}}

\def\CI{{\mathfrak{CI}}}

\def\GHU{\textup{\textsf{GHU}}}
\def\NSC{\textup{\textsf{NSC}}}
\def\ST{\textup{\textsf{ST}}}
\def\MPC{\textup{\textsf{MPC}}}
\def\sigtoe{\Sigma_{2}^{0}\textup{\textsf{-TRANS}}}
\def\GH{\textup{\textsf{GH}}}

\def\SCF{\textup{\textsf{SCF}}}

\def\DIV{\textup{\textsf{DIV}}}

\def\PST{\textup{\textsf{PST}}}
\def\her{\textup{\textsf{her}}}

\def\MU{\textup{\textsf{MU}}}
\def\MUO{\textup{\textsf{MUO}}}

\def\CRI{\textup{\textsf{CRI}}}

\def\UDNR{\textup{\textsf{UDNR}}}

\numberwithin{equation}{section}
\numberwithin{thm}{section}

\title{	The computational content of Nonstandard Analysis}
\author{Sam Sanders
\institute{MCMP\\ LMU Munich, Germany}
\institute{Department of Mathematics\\
Ghent University\\
Ghent, Belgium}
\email{sasander@me.com}
}
\def\titlerunning{The computational content of Nonstandard Analysis}
\def\authorrunning{S.\ Sanders}
\begin{document}
\maketitle

\begin{abstract}
Kohlenbach's \emph{proof mining} program deals with the extraction of effective information from typically ineffective proofs. Proof mining has its roots in Kreisel's pioneering work on the so-called unwinding of proofs.   The proof mining of classical mathematics is rather restricted in scope due to the existence of sentences without computational content which are provable from the law of excluded middle and which involve \emph{only two} quantifier alternations.  By contrast, we show that the proof mining of classical Nonstandard Analysis has a very large scope. In particular, we will observe that this scope includes any theorem of \emph{pure} Nonstandard Analysis, where `pure' means that only nonstandard definitions (and not the \emph{epsilon-delta} kind) are used. In this note, we survey results in analysis, computability theory, and Reverse Mathematics. 
%
\end{abstract}


\section{Introduction}\label{intro}
\noindent
The aim of this note is to survey the vast \emph{computational} content of classical Nonstandard Analysis as established in \cite{samzoo, samzooII, samGH, sambon}.
Results are mostly presented without proofs but references are provided.  

\medskip

First of all, numerous practitioners of Nonstandard Analysis have alluded to the constructive nature of its praxis; The following quotes serve as a representative illustration.  
\begin{quote}
\emph{It has often been held that nonstandard analysis is highly non-constructive, thus somewhat suspect, depending as it does upon the ultrapower construction to produce a model \textup{[\dots]} On the other hand, nonstandard \emph{praxis} is remarkably constructive; having the extended number set we can proceed with explicit calculations.} (Emphasis in original: \cite{NORSNSA}*{p.\ 31})
\end{quote}
\begin{quote}
\emph{Those who use nonstandard arguments often say of their proofs that they are ``constructive modulo an ultrafilter''; implicit in this statement is the suggestion that such arguments might give rise to genuine constructions}. (\cite{rossenaap}*{p.\ 494})
\end{quote}
The reader may interpret the word \emph{constructive} as the mainstream/classical notion `effective', or as the foundational notion from Bishop's \emph{Constructive Analysis} (\cite{bridge1}).  
As will become clear, both cases will be treated below \emph{and separated carefully}.      

\medskip

To uncover the computational content of Nonstandard Analysis alluded to in the above quotes, we shall introduce a template $\CI$ in Section \ref{detail} which converts a theorem of \emph{pure} Nonstandard Analysis into the associated `constructive' theorem;   
Here, a theorem of `pure' Nonstandard Analysis is one formulated solely with the \emph{nonstandard} definitions (of continuity, convergence, etc) rather than the usual `epsilon-delta' definitions.  
We present a wide range of applications of the template $\CI$ in this note.  

\medskip

On a historical note, the late Grigori Mints has repeatedly pushed the author to investigate the computational content of classical Nonstandard Analysis.  
In particular, Mints conjectured the existence of results analogous or similar to Kohlenbach's \emph{proof mining} program (\cite{kohlenbach3}).  
The latter program has its roots in Kreisel's pioneering work on the `unwinding' of proofs, where the latter's goal is similar to ours:
\begin{quote}
\emph{To determine the constructive \(recursive\) content or the constructive equivalent of the non-constructive concepts and theorems used in mathematics}, particularly arithmetic and analysis.  (Emphasis in original on \cite{kreimiearivier}*{p.\ 155})
\end{quote}
%
%
Finally, Horst Osswald has qualified the observation from the above quotes as \emph{Nonstandard Analysis is locally constructive}, to be understood as the fact that the mathematics performed in the nonstandard world is highly effective while the principles needed to `jump between' the nonstandard world and usual mathematics, are highly non-constructive in general (See \cite{nsawork2}*{\S7}, \cite{Oss3}*{\S1-2}, or \cite{Oss2}*{\S17.5}).  The results in this paper shall be seen to vindicate both the Mints and Osswald view.  

\section{About and around internal set theory}
In this section, we introduce Nelson's \emph{internal set theory}, and its fragments $\P$ and $\H$ from \cite{brie}.  
We discuss the term extraction result in Corollary \ref{consresultcor}, which is central to our enterprise.  
\subsection{Internal set theory}\label{PIPI}
In Nelson's \emph{syntactic} approach to Nonstandard Analysis (\cite{wownelly}), as opposed to Robinson's semantic one (\cite{robinson1}), a new predicate `st($x$)', read as `$x$ is standard' is added to the language of \textsf{ZFC}, the usual foundation of mathematics.  
The notations $(\forall^{\st}x)$ and $(\exists^{\st}y)$ are short for $(\forall x)(\st(x)\di \dots)$ and $(\exists y)(\st(y)\wedge \dots)$.  A formula is \emph{internal} if it does not involve `st', and \emph{external} otherwise.   
The three external axioms \emph{Idealisation}, \emph{Standard Part}, and \emph{Transfer} govern the new predicate `st';  They are respectively defined\footnote{The superscript `fin' in \textsf{(I)} means that $x$ is finite, i.e.\ its number of elements are bounded by a natural number.} as:  
\begin{enumerate}
\item[\textsf{(I)}] $(\forall^{\st~\textup{fin}}x)(\exists y)(\forall z\in x)\varphi(z,y)\di (\exists y)(\forall^{\st}x)\varphi(x,y)$, for internal $\varphi$ with any parameters.  
\item[\textsf{(S)}] $(\forall^{\st} x)(\exists^{\st}y)(\forall^{\st}z)\big((z\in x\wedge \varphi(z))\asa z\in y\big)$, for any $\varphi$.
\item[\textsf{(T)}] $(\forall^{\st}t)\big[(\forall^{\st}x)\varphi(x, t)\di (\forall x)\varphi(x, t)\big]$, where $\varphi(x,t)$ is internal, and only has free variables $t, x$.  
\end{enumerate}
The system \textsf{IST} is (the internal system) \textsf{ZFC} extended with the aforementioned external axioms;  
The former is a conservative extension of \textsf{ZFC} for the internal language, as proved in \cite{wownelly}.    

\medskip

In \cite{brie}, the authors study G\"odel's system $\textsf{T}$ extended with special cases of the external axioms of \textsf{IST}.  
In particular, they introduce the systems $\H$ and $\P$ which are conservative extensions of the (internal) logical systems \textsf{E-HA}$^{\omega}$ and $\textsf{E-PA}^{\omega}$, respectively \emph{Heyting and Peano arithmetic in all finite types and the axiom of extensionality}.       
We refer to \cite{kohlenbach3}*{\S3.3} and \cite{brie}*{\S2} for the exact definitions of the (mainstream in mathematical logic) systems \textsf{E-HA}$^{\omega}$ and $\textsf{E-PA}^{\omega}$ and the associated extensions \textsf{E-HA}$^{\omega*}$ and $\textsf{E-PA}^{\omega*}$.  
We refer to \cite{sambon} and \cite{samzooII} for the exact definition of the systems $\P$ and $\H$.  Their importance lies in the \emph{term extraction corollary} which we discuss in the next section.

\medskip

Note that the contraposition of the idealisation axiom \textsf{(I)} allows one to `push outside' a standard quantifier.  The axiom \textsf{(I)} (formulated in the language of finite types) is included in $\P$ and $\H$; We shall need this axiom in the proof of Theorem \ref{varou} and therefore list it as follows:  
\bdefi[Idealisation \textsf{I}]
For any internal formula $\varphi$, we have
\be\label{pislke}
(\forall^{\st} x^{\sigma^{*}})(\exists y^{\tau} )(\forall z^{\sigma}\in x)\varphi(z,y)\di (\exists y^{\tau})(\forall^{\st} x^{\sigma})\varphi(x,y), 
\ee
\edefi
\noindent Note that $x^{\sigma^{*}}$ in the antecedent of \eqref{pislke} is a finite sequence of objects of type $\sigma$.

\medskip

Finally, we note that $\IST$ is just $\ZFC$ with an extra unary predicate governed by the aforementioned axioms.  In other words, all the \emph{usual} definitions (of real function, large cardinal, Turing machine, etc) from $\ZFC$ can also be stated in $\IST$ \emph{by exactly the same formula} of $\ZFC$.  The same holds for $\P$ and $\H$; In particular, the latter systems use Kohlenbach's definition of  real number and real function from \cite{kohlenbach2} in the higher-type framework.

\subsection{The term extraction corollary}\label{kintro}
In this section, we discuss the central tool of our investigation, namely the \emph{term extraction corollary} of the system $\P$, and sketch its vast scope.  
The following is essentially a corollary to \cite{brie}*{Theorem 7.7}.    
\begin{cor}[Term extraction]\label{consresultcor}
If $\Delta_{\intern}$ is a collection of internal formulas and $\psi$ is internal, and
\be\label{bog}
\P + \Delta_{\intern} \vdash (\forall^{\st}\underline{x})(\exists^{\st}\underline{y})\psi(\underline{x},\underline{y}, \underline{a}), 
\ee
then one can extract from the proof a sequence of closed terms $t$ in $\mathcal{T}^{*}$ such that
\be\label{effewachten}
\textup{\textsf{E-PA}}^{\omega*}+ \Delta_{\intern} \vdash (\forall \underline{x})(\exists \underline{y}\in t(\underline{x}))\psi(\underline{x},\underline{y},\underline{a}).
\ee
\end{cor}
Note that $t$ does not provide a witnessing functional for $(\exists y)$ in \eqref{effewachten}; In particular  $t(x)$ is only a \emph{finite sequence} (of length $ |t(x)|$) of witnesses for $(\exists y)$.  
For the remainder of this paper, the notion of `normal form' shall \emph{always} refer to a formula of the form $(\forall^{\st}x)(\exists^{\st}y)\varphi(x, y)$ with $\varphi$ internal, i.e.\ without `st'. 

\medskip

Curiously, the previous corollary is not proved in \cite{brie};  A proof making essential use of \cite{brie}*{Theorem~7.7} may be found in \cites{samzoo,sambon}.  
The previous corollary \emph{is} proved for the constructive system $\H$ rather than the classical system $\P$ in \cite{brie}*{Theorem 5.9}, but our interest goes out to classical systems.  
Furthermore, Corollary \ref{consresultcor} does not depend on the full strength of Peano arithmetic:  The same result holds for any system which at least includes $\textsf{EFA}$, also called $I\Delta_{0}+\textsf{EXP}$.   

\medskip

Clearly, Corollary \ref{consresultcor} allows us to extract effective information (in the form of the term $t$) from proofs as in \eqref{bog} in Nonstandard Analysis, to obtain effective results as in \eqref{effewachten} \emph{not involving Nonstandard Analysis}.   
We now discuss why Corollary~\ref{effewachten} has such a vast scope, including all of `pure' Nonstandard Analysis, as claimed in the introduction.  

%
\begin{enumerate}
\item First of all, the \emph{nonstandard} definitions of common notions (such as continuity, integrability, convergence, compactness, et cetera) in Nonstandard Analysis can be brought into the `normal form' $(\forall^{\st}x)(\exists^{\st}y)\varphi(x, y)$.  This can always be done in $\P$ and usually in $\H$.   As an example, nonstandard continuity for $f:\R\di \R$ as follows (where `$x\approx y$' is short for $(\forall^{\st}n^{0})(|x-y|<_{\R}\frac1n)$):
\be\label{soareyou3}
(\forall^{\st}x\in [0,1])(\forall y\in [0,1])[x\approx y \di f(x)\approx f(y)].
\ee
is \emph{equivalent} (over $\P$ or $\H$) to the following normal form by Theorem \ref{helful}:
\be\label{soareyou3098}\textstyle
(\forall^{\st}k^{0})(\forall^{\st}x\in [0,1])(\exists^{\st}N^{0}) \underline{(\forall y\in [0,1])(|x- y|<_{\R}\frac{1}{N} \di |f(x)- f(y)|<_{\R}\frac1k}\big),
\ee
where the underlined formula is internal.  
Similar equivalences hold for nonstandard definitions of compactness, Riemann integration, differentiability, convergence, et cetera.  


\item Secondly, the normal forms are closed under \emph{modus ponens}: Indeed, it is possible (easy in $\P$ and involved in $\H$) to show that an implication between normal forms:
\[
(\forall^{\st}x_{0})(\exists^{\st}y_{0})\varphi_{0}(x_{0}, y_{0})\di (\forall^{\st}x_{1})(\exists^{\st}y_{1})\varphi_{1}(x_{1}, y_{1}),
\]
can \emph{also} be brought into a normal form $(\forall^{\st}x)(\exists^{\st}y)\varphi(x, y)$.  Hence, it seems theorems of `pure' Nonstandard Analysis, i.e.\ formulated solely with nonstandard definitions like \eqref{soareyou3}, can be brought into the latter normal form. 
This can always be done in $\H$ and $\P$.     

\item Third, normal forms are closed under quantification over the nonstandard numbers:  In particular, for internal formulas $\varphi(x, y, M)$, the following formula 
\[
(\forall M^{0})\big(\neg\st(M)\di (\forall^{\st}x))(\exists^{\st}y)\varphi(x,y, M)\big), 
\]
is equivalent to a normal form in $\P$.  The same holds for quantification over nonstandard \emph{higher-type} objects.  This item is significant because applications of Nonstandard Analysis often start with `divide the compact space at hand into pieces of infinitesimal surface/volume/measure $\frac{1}{M^{0}}$'.   

\item Fourth, the normal form $(\forall^{\st}x)(\exists^{\st}y)\varphi(x, y)$ has exactly the right structure to yield the effective version $(\forall x)\varphi(x, t(x))$.   
In particular, from the proof of the normal form $(\forall^{\st}x)(\exists^{\st}y)\varphi(x, y)$ (inside $\H$ or $\P$), a term $s$ can be `read off' such that $(\forall x)(\exists y\in s(x))\varphi(x, y)$ 
has a proof inside a system \emph{involving no Nonstandard Analysis}.  The term $t$ is then defined in terms of $s$.  For theorems of analysis, $\varphi(x,y)$ is often monotone in $y$, and $t$ is then just the maximum of all entries of $s$.  For instance, \eqref{soareyou3098} is `monotone in $N$' in the sense that any larger number than $N$ will also do.
\end{enumerate}
It is important to note that there is \emph{no general procedure} to convert the `weak witnessing' term $s$ into a `strong witnessing' term $t$ in the fourth step.  However, when dealing with mathematical theorems (rather than purely logical statements), experience bears out that this conversion is almost always possible.    

\medskip

It goes without saying that most technical details have been omitted from the above sketch, this in order to promote intuitive understanding.  
Nonetheless, the previous four steps form the skeleton of the template $\CI$ introduced in Section~\ref{detail}.  
In light of the previous observations, the class of normal forms, and hence the scope of $\CI$, seems to be \emph{very large}, which is what we intend to establish in the remainder.  
 
\medskip

Finally, the results in this note should be contrasted with the `mainstream' view of Nonstandard Analysis:  
One usually thinks of the universe of standard objects as `the usual world of mathematics', which can be studied `from the outside' using nonstandard objects such as infinitesimals.     
In this richer framework, proofs can be much shorter than those from standard (=non-Nonstandard) analysis;  Furthermore, there are \emph{conservation results} guaranteeing that theorems of usual mathematics \emph{proved using Nonstandard Analysis} can also be proved \emph{without using Nonstandard Analysis}.  Thus, the starting and end point (according to the mainstream view) is always the universe of standard objects, i.e.\ usual mathematics.  By contrast, our starting point is \emph{pure} Nonstandard Analysis and our end point is effective mathematics.

\section{An elementary example}
In this section, we present an elementary example which we believe to be enlightening.  
Based on this example, we formulate a general template $\CI$ to obtain effective theorems from nonstandard ones.  
\subsection{From continuity to Riemann integration}
In this section, we study the statement $\CRI$: \emph{A uniformly continuous function on the unit interval is Riemann integrable}.  
We first obtain the effective version of $\CRI$ from the nonstandard version inside $\P$.  
We then obtain the same result in the \emph{constructive} system $\H$.  Finally, we re-obtain the nonstandard version from a special effective version, called the \emph{Hebrandisation}.

\medskip

First of all, the `usual' nonstandard definitions of continuity and integration are as follows.  Recall that `$x\approx y$' is an abbreviation for `$(\forall^{\st}n)(|x-y|<_{\R}{\frac{1}{n}})$'.  
\bdefi[Continuity]\label{Kont}
A function $f:\R\di \R$ is \emph{nonstandard continuous} on $[0,1]$ if
\be\label{soareyou32}
(\forall^{\st}x\in [0,1])(\forall y\in [0,1])[x\approx y \di f(x)\approx f(y)].
\ee
A function $f:\R\di \R$ is \emph{nonstandard uniformly continuous} on $[0,1]$ if
\be\label{soareyou4}
(\forall x, y\in [0,1])[x\approx y \di f(x)\approx f(y)].
\ee
\edefi
\bdefi[Integration]\label{kunko}~
\begin{enumerate}
\item A \emph{partition} of $[0,1]$ is any sequence $\pi=(0, t_{0}, x_{1},t_{1},  \dots,x_{M-1}, t_{M-1}, 1)$.  We write `$\pi \in P([0,1]) $' to denote that $\pi$ is such a partition.
\item For $\pi\in P([0,1])$, $\|\pi\|$ is the \emph{mesh}, i.e.\ the largest distance between two partition points $x_{i}$ and $x_{i+1}$. 
\item For $\pi\in P([0,1])$ and $f:\R\di \R$, $S_{\pi}(f):=\sum_{i=0}^{M-1}f(t_{i}) (x_{i}-x_{i+1}) $ is the \emph{Riemann sum} of $f$ and $\pi$.  
\item A function $f:\R\di \R$ is \emph{nonstandard Riemann integrable} on $[0,1]$ if
\be\label{soareyou5}
(\forall \pi, \pi' \in P([0,1]))\big[\|\pi\|,\| \pi'\|\approx 0  \di S_{\pi}(f)\approx S_{\pi}(f)  \big].
\ee
\end{enumerate}
\edefi
Secondly, it was claimed in the previous section that nonstandard continuity has a nice normal form. 
\begin{thm}[$\P$]\label{helful}
Nonstandard uniform continuity \eqref{soareyou4} is equivalent to 
\be\label{soareyou301}\textstyle
(\forall^{\st}k^{0})(\exists^{\st}N^{0}) {(\forall x,y\in [0,1])(|x- y|<_{\R}\frac{1}{N} \di |f(x)- f(y)|<_{\R}\frac1k}\big),
\ee  
\end{thm}
\begin{proof}
We only need to prove the forward implication.  
Resolving `$\approx$' in \eqref{soareyou4}, we obtain 
\[\textstyle
{(\forall x,y\in [0,1])((\forall ^{\st}N^{0}) |x- y|<_{\R}\frac{1}{N} \di (\forall^{\st}k)|f(x)- f(y)|<_{\R}\frac1k}\big),
\]
and pushing outside all standard quantifiers, we obtain 
\[\textstyle
(\forall^{\st}k^{0}) {(\forall x,y\in [0,1])(\exists^{\st}N^{0})\underline{(|x- y|<_{\R}\frac{1}{N} \di |f(x)- f(y)|<_{\R}\frac1k}\big)},
\]
where the underlined formula is internal.  Applying the contraposition of idealisation \textsf{I} as in \eqref{pislke}: 
\[\textstyle
(\forall^{\st}k^{0})(\exists^{\st}w^{0^{*}}) {(\forall x,y\in [0,1])(\exists N^{0}\in w){(|x- y|<_{\R}\frac{1}{N} \di |f(x)- f(y)|<_{\R}\frac1k}\big)}.
\]
Now let $M$ be the maximum of all numbers in $w=(n_{0}^{0}, \dots, n_{k}^{0})$, and note that 
\[\textstyle
(\forall^{\st}k^{0})(\exists^{\st}M) {(\forall x,y\in [0,1]){(|x- y|<_{\R}\frac{1}{M} \di |f(x)- f(y)|<_{\R}\frac1k}\big)},
\]
by the monotonicity of the internal formula.  This is exactly \eqref{soareyou301}, and we are done.  
\end{proof}
Thirdly, we now introduce the nonstandard and effective versions of $\CRI$ as follows.
\begin{thm}[$\CRI_{\ns}$]
Every nonstandard uniformly continuous function on the unit interval is nonstandard Riemann integrable there.  
\end{thm}
\begin{thm}[$\CRI_{\ef}(t)$]
For any $f:\R\di\R$ with modulus of uniform continuity $g$, the functional $t(g)$ is a modulus of Riemann integration, i.e.\ we have
\begin{align}\textstyle
(\forall \textstyle x, y \in [0,1],k)(|x-y|<\frac{1}{g(k)}&\textstyle \di |f(x)-f(y)|\leq\frac{1}{k})\label{EST}\\
&\textstyle\di  (\forall n)(\forall \pi, \pi' \in P([0,1]))\big(\|\pi\|,\| \pi'\|< \frac{1}{t(g)(n)}  \di |S_{\pi}(f)- S_{\pi'}(f)|\leq \frac{1}{n} \big).\notag
\end{align}
\end{thm}
Kohlenbach has shown that continuous real-valued functions as represented in RM (See \cite{simpson2}*{II.6.6} and \cite{kohlenbach3}*{Prop.\ 4.4}) have a modulus of (pointwise) continuity as in the antecedent of \eqref{EST}.  
\begin{thm}\label{varou}
From a proof of $\CRI_{\ns}$ in $\P$, a term $t$ can be extracted such that $\textup{\textsf{E-PA}}^{\omega*} $ proves $\CRI_{{\ef}}(t)$.  
\end{thm}
\begin{proof}
The theorem $\CRI_{\ns}$ can be proved in far weaker systems than $\P$ by \cite{aloneatlast3}*{Theorem 19}.
A variation of the latter proof may be found in \cite{sambon}*{\S3.1.1}.  We now sketch how to obtain a normal form for $\CRI_{\ns}$.   
Applying Corollary \ref{consresultcor} to this normal form will then yield $\CRI_{\ef}(t)$.  

\medskip

First of all, a normal form for uniform nonstandard continuity \eqref{soareyou4} is \eqref{soareyou301},  
while the (equivalent) normal form for nonstandard Riemann integration similarly is:  
\be\label{soareyou57}\textstyle
(\forall^{\st}n^{0})(\exists^{\st}M^{0})(\forall \pi, \pi' \in P([0,1]))\big[\|\pi\|,\| \pi'\|<\frac{1}{M}  \di |S_{\pi}(f)- S_{\pi}(f) |<_{\R}\frac1k\big].
\ee  
Secondly, in light of the previous equivalences, $\CRI_{\ns}$ is the implication $\eqref{soareyou301}\di \eqref{soareyou57}$ for all $f:\R\di \R$.  By strengthening the antecedent of the latter implication, we obtain for all $f:\R\di \R$ and all \emph{standard} $g$: 
\begin{align}\textstyle
(\forall^{\st}k)(\forall &\textstyle x, y \in [0,1])(|x-y|<\frac{1}{g(k)}\textstyle \di |f(x)-f(y)|\leq\frac{1}{k})\label{ES2T}\\
&\textstyle\di  (\forall^{\st} n)(\exists^{\st}M)(\forall \pi, \pi' \in P([0,1]))\big(\|\pi\|,\| \pi'\|< \frac{1}{M}  \di |S_{\pi}(f)- S_{\pi'}(f)|\leq \frac{1}{n} \big).\notag
\end{align}
Now drop the `st' in the `$(\forall^{\st}k)$' quantifier in \eqref{ES2T}, and bring outside all standard quantifiers to obtain:
\begin{align}\textstyle
 (\forall^{\st} n, g)(\forall f:\R\di \R)(\exists^{\st}M)\Big[(\forall k,  &\textstyle x, y \in [0,1])(|x-y|<\frac{1}{g(k)}\textstyle \di |f(x)-f(y)|\leq\frac{1}{k})\label{E32T}\\
&\textstyle\di (\forall \pi, \pi' \in P([0,1]))\big(\|\pi\|,\| \pi'\|< \frac{1}{M}  \di |S_{\pi}(f)- S_{\pi'}(f)|\leq \frac{1}{n} \big)\Big].\notag
\end{align}
Applying idealisation \textsf{(I)}, we obtain that: 
\begin{align}\textstyle
 (\forall^{\st} n, g)(\exists^{\st}w)(\forall f:\R\di \R)(\exists M\in w)\Big[(\forall k,  &\textstyle x, y \in [0,1])(|x-y|<\frac{1}{g(k)}\textstyle \di |f(x)-f(y)|\leq\frac{1}{k})\label{E312T}\\
&\textstyle\di (\forall \pi, \pi' \in P([0,1]))\big(\|\pi\|,\| \pi'\|< \frac{1}{M}  \di |S_{\pi}(f)- S_{\pi'}(f)|\leq \frac{1}{n} \big)\Big].\notag
\end{align}
Now let $N$ be the maximum of all numbers in $w$ from \eqref{E312T}, and note that 
\begin{align}\textstyle
 (\forall^{\st} n, g)(\exists^{\st}N)(\forall f:\R\di \R)\Big[(\forall k,  &\textstyle x, y \in [0,1])(|x-y|<\frac{1}{g(k)}\textstyle \di |f(x)-f(y)|\leq\frac{1}{k})\label{RF2T}\\
&\textstyle\di (\forall \pi, \pi' \in P([0,1]))\big(\|\pi\|,\| \pi'\|< \frac{1}{N}  \di |S_{\pi}(f)- S_{\pi'}(f)|\leq \frac{1}{n} \big)\Big],\notag
\end{align}
due to the monotone behaviour of the consequent.  Now apply the term extraction corollary to `$\P\vdash \eqref{RF2T}$' to obtain a term $s$ such that \textsf{E-PA}$^{\omega*}$ proves
\begin{align}\textstyle
 (\forall n, g)(\exists N\in s(g,n))(\forall f:\R\di \R)\Big[(\forall k,  &\textstyle x, y \in [0,1])(|x-y|<\frac{1}{g(k)}\textstyle \di |f(x)-f(y)|\leq\frac{1}{k})\label{RF12T}\\
&\textstyle\di (\forall \pi, \pi' \in P([0,1]))\big(\|\pi\|,\| \pi'\|< \frac{1}{N}  \di |S_{\pi}(f)- S_{\pi'}(f)|\leq \frac{1}{n} \big)\Big],\notag
\end{align}
Define $t(g,n)$ as the maximum number of $s(g, n)$, and note that \eqref{RF12T} implies $\CRI_{\ef}(t)$, again due to the monotone behaviour of the consequent.  
\end{proof}
\begin{cor}\label{cordejedi}
Theorem \ref{varou} also goes through constructively, i.e.\ we can prove $\CRI_{\ns}$ in $\textup{\textsf{H}} $ and a term $t$ can be extracted such that $\textup{\textsf{E-HA}}^{\omega*} $ proves $\CRI_{{\ef}}(t)$.  
\end{cor}
\begin{proof}
The proof of $\CRI_{\ns}$ in \cite{aloneatlast3} is clearly constructive in the sense of $\H$.   A careful inspection of the proof of the theorem shows that \eqref{RF2T} can also be derived 
in $\H$ from $\CRI_{\ns}$.  Applying the term extraction result for $\H$ (\cite{brie}*{Theorem 5.6}) then yields the corollary.   A full proof is in \cite{sambon}*{\S3.1}.
\end{proof}
Finally, define the \emph{Herbrandisation} of $\CRI_{\ns}$ as follows:
\begin{align}\textstyle
\textstyle~(\forall f, g ,k')\Big[(\forall k\leq s(g,k'))&(\forall \textstyle x, y \in [0,1])(|x-y|<\frac{1}{g(k)} \di |f(x)-f(y)|\leq\frac{1}{k})\tag{$\CRI_{\her}(s,t)$}\\
&\label{FEST}\textstyle\di  (\forall \pi, \pi' \in P([0,1]))\big(\|\pi\|,\| \pi'\|< \frac{1}{t(g,k')}  \di |S_{\pi}(f)- S_{\pi}(f)|\leq \frac{1}{k'} \big)  \Big]\notag
\end{align}
The Herbrandisation $\CRI_{\her}(s,t)$ follows from $\CRI_{\ns}$ in the same way as in the theorem.  In particular, we obtain the former if we do not drop the `st' in `$(\forall^{\st}k)$' to obtain \eqref{E32T}.   
We have the following corollary.
\begin{cor}\label{somaar} Let $t$ be a term in the internal language. 
A proof inside $\textup{\textsf{E-PA}}^{\omega*}$ of the Herbrandisation $\CRI_{\her}(s,t)$, can be converted into a proof inside $\P$ of $\CRI_{\ns}$.  
\end{cor}
\begin{proof}
The basic axioms of $\P$ state that any term of the internal language is standard.  The rest of the corollary is now straightforward.  A full proof is in \cite{sambon}*{\S3.1}.  
\end{proof}

\subsection{The template $\CI$}\label{detail}
In this section, we formulate the template $\CI$ based on the above case study.  
We emphasize that some aspects of $\CI$ are inherently vague.  Recall from the previous section that a `normal form' is a formula of the form $(\forall^{\st}x)(\exists^{\st}y)\varphi(x, y)$ with $\varphi$ internal.    
\begin{tempie}[$\CI$]~\rm
The starting point for $\CI$ is a theorem $T$ formulated in the language of $\textsf{E-PA}^{\omega*}$.  
\begin{enumerate}[(i)]
\item Replace in $T$ all definitions (convergence, continuity, et cetera) by their well-known counterparts from Nonstandard Analysis.  For the resulting theorem $T^{*}$, look up the proof (e.g.\ in 
\cites{loeb1,stroyan, nsawork2}) and formulate it inside $\P$ or $\H$ if possible.   
If $T^{*}$ cannot be proved in $\P$, consider $A\di T^{*}$, where $A$ is a collection of external axioms from $\IST$ to guarantee the provability in $\P$.        \label{famouslastname}
\item Bring all nonstandard definitions in $T^{*}$ into the \emph{normal form} $(\forall^{\st}x)(\exists^{\st}y)\varphi(x, y)$.  This operation usually requires $\textsf{I}$ for $\P$, and usually requires extra axioms for $\H$.    
If necessary, drop `st' in leading existential quantifiers of positively occurring formulas (like to obtain \eqref{E32T}).  \label{famousfirsdtname}
\item Starting with the most deeply nested implication, bring
\be\label{duggg}
(\forall^{\st}x_{0})(\exists^{\st}y_{0})\varphi_{0}(x_{0}, y_{0})\di (\forall^{\st}x_{1})(\exists^{\st}y_{1})\varphi_{1}(x_{1}, y_{1}),
\ee
into a normal form $(\forall^{\st}x)(\exists^{\st}y)\varphi(x, y)$.  
\item Apply Corollary \ref{consresultcor} (if applicable \cite{brie}*{Theorem 5.6} for $\H$) to the proof of the normal form of $T^{*}$.  
\item Output the term(s) $t$ and the proof(s) of the effective version.  Modify these terms for monotone formulas if necessary.  
\end{enumerate}    
\end{tempie}
The theorems in the above case study all had proofs inside $\H$ or $\P$, i.e.\ the final sentence in step \eqref{famouslastname} does not apply.  
In Section \ref{RMSTUD}, we shall study theorems for which we \emph{do} have to add external axioms of $\IST$ to the conditions of the theorem.  

\medskip

Finally, there is a tradition of Nonstandard Analysis in RM and related topics (See e.g.\ \cites{pimpson,tahaar, tanaka1, tanaka2, horihata1, yo1, yokoyama2, yokoyama3}), which provides a source of proofs in (pure) Nonstandard Analysis for $\CI$.  To automate the process of applying $\CI$, we have initiated the implementation of the term extraction algorithm from Corollary \ref{consresultcor} in Agda, which is work in progress at this time (\cite{EXCESS}).

\section{Reverse Mathematics}\label{RM}
In this section, we first introduce the program \emph{Reverse Mathematics}, and then list results regarding the main systems consider therein.   
\subsection{Introducing Reverse Mathematics}
Reverse Mathematics (RM) is a program in the foundations of mathematics initiated around 1975 by Friedman (\cites{fried,fried2}) and developed extensively by Simpson (\cite{simpson2, simpson1}) and others.  
The aim of RM is to find the axioms necessary to prove a statement of \emph{ordinary} mathematics, i.e.\ dealing with countable or separable spaces.   
The classical\footnote{In \emph{Constructive Reverse Mathematics} (\cite{ishi1}), the base theory is based on intuitionistic logic.} base theory $\RCA_{0}$ of `computable\footnote{The system $\RCA_{0}$ consists of induction $I\Sigma_{1}$, and the {\bf r}ecursive {\bf c}omprehension {\bf a}xiom $\Delta_{1}^{0}$-CA.} mathematics' is always assumed.  
Thus, the aim of RM is as follows:  
\begin{quote}
\emph{The aim of \emph{RM} is to find the minimal axioms $A$ such that $\RCA_{0}$ proves $ [A\di T]$ for statements $T$ of ordinary mathematics.}
\end{quote}
Surprisingly, once the minimal axioms $A$ have been found, we almost always also have $\RCA_{0}\vdash [A\asa T]$, i.e.\ not only can we derive the theorem $T$ from the axioms $A$ (the `usual' way of doing mathematics), we can also derive the axiom $A$ from the theorem $T$ (the `reverse' way of doing mathematics).  In light of the latter, the field was baptised `Reverse Mathematics'.    

\medskip

Perhaps even more surprisingly, in the majority\footnote{Exceptions are classified in the so-called Reverse Mathematics Zoo (\cite{damirzoo}).
} 
of cases for a statement $T$ of ordinary mathematics, either $T$ is provable in $\RCA_{0}$, or the latter proves $T\asa A_{i}$, where $A_{i}$ is one of the logical systems $\WKL_{0}, \ACA_{0},$ $ \ATR_{0}$ or $\FIVE$.  The latter together with $\RCA_{0}$ form the `Big Five' and the aforementioned observation that most mathematical theorems fall into one of the Big Five categories, is called the \emph{Big Five phenomenon} (\cite{montahue}*{p.~432}).  
Furthermore, each of the Big Five has a natural formulation in terms of (Turing) computability (See e.g.\ \cite{simpson2}*{I.3.4, I.5.4, I.7.5}).
As noted by Simpson in \cite{simpson2}*{I.12}, each of the Big Five also corresponds (sometimes loosely) to a foundational program in mathematics.  

\medskip

The logical framework for Reverse Mathematics is \emph{second-order arithmetic}, in which only natural numbers and sets thereof are available.  As a result, functions from reals to reals are not available, and have to be represented by so-called \emph{codes} (See \cite{simpson2}*{II.6.1}).  In the latter case, the coding of continuous functions amounts to introducing a modulus of (pointwise) continuity (See \cite{kohlenbach4}*{\S4}).  The nonstandard theorems proved in the system $\P$ \emph{will not involve coding} (for continuous functions or otherwise); 
However, as will become clear below, a modulus of continuity naturally `falls out of' the nonstandard definition of continuity as in \eqref{soareyou3}.  
Thus, the nonstandard framework seems to `do the coding for us'.

\medskip

In light of the previous, one of the main results of RM is that \emph{mathematical theorems} fall into \emph{only five} logical categories.  
By contrast, there are lots and lots of (purely logical or non-mathematical) statements which fall outside of these five categories.  Similarly, most \emph{mathematical theorems} from Nonstandard Analysis 
have the normal from required for applying term extraction via Corollary~\ref{consresultcor}, while there are plenty of non-mathematical or purely logical statements which do not.  
In conclusion, the results in this paper are inspired by the \emph{Reverse Mathematics way of thinking} that mathematical theorems (known in the literature) will behave `much nicer' than arbitrary formulas (even of restricted complexity).  In particular, since there is no meta-theorem for the (Big Five and its zoo) classification of RM, one cannot hope to obtain a meta-theorem for the template $\CI$ from Section~\ref{detail}.     

\subsection{The Big Five}\label{RMSTUD}
In this section, we list the results of applying $\CI$ to equivalences involving the strongest three Big Five systems.  
We do not go into the details regarding $\WKL_{0}$ because of a lack of space.  Proofs may be found in \cite{sambon}*{\S4}.  
%
%

\subsubsection{Theorems equivalent to $\ACA_{0}$}\label{AZA}
In this section, we study the \emph{monotone convergence theorem} $\MCT$, i.e.\ the statement that \emph{every bounded increasing sequence of reals is convergent}, which is equivalent to arithmetical comprehension $\ACA_{0}$ by \cite{simpson2}*{III.2.2}.
We prove an equivalence between a nonstandard version of $\MCT$ and a fragment of \emph{Transfer}.
From this nonstandard equivalence, we obtain an effective RM equivalence involving $\MCT$ and arithmetical comprehension by applying $\CI$.  

\medskip

Firstly, the nonstandard version of $\MCT$ (involving nonstandard convergence) is:
\be\label{MCTSTAR}\tag{\MCT$_{\textsf{ns}}$}
(\forall^{\st} c_{(\cdot)}^{0\di 1})\big[(\forall n^{0})(c_{n}\leq c_{n+1}\leq 1)\di (\forall N,M\in \Omega)[c_{M}\approx c_{N}]    \big], 
\ee
where `$(\forall K\in \Omega)(\dots)$' is short for $(\forall K^{0}))(\neg\st(K)\di \dots)$.  
The effective version \MCT$_{\textsf{ef}}(t)$ is:
\be\label{MCTSTAR22}\textstyle
(\forall c_{(\cdot)}^{0\di 1},k^{0})\big[(\forall n^{0})(c_{n}\leq c_{n+1}\leq 1)\di (\forall N,M\geq t(c_{(\cdot)})(k))[|c_{M}- c_{N}|\leq \frac{1}{k} ]   \big].
\ee
We require two equivalent (\cite{kohlenbach2}*{Prop.\ 3.9}) versions of arithmetical comprehension: 
\be\label{mu}\tag{$\mu^{2}$}
(\exists \mu^{2})\big[(\forall f^{1})( (\exists n)f(n)=0 \di f(\mu(f))=0)    \big],
\ee
\be\label{mukio}\tag{$\exists^{2}$}
(\exists \varphi^{2})\big[(\forall f^{1})( (\exists n)f(n)=0 \asa \varphi(f)=0   ) \big],
\ee
and also the restriction of Nelson's axiom \emph{Transfer} as follows:
\be\tag{$\paai$}
(\forall^{\st}f^{1})\big[(\forall^{\st}n^{0})f(n)\ne0\di (\forall m)f(m)\ne0\big].
\ee
Denote by $\textsf{MU}(\mu)$ the formula in square brackets in \eqref{mu}.  We have the following theorem which establishes the explicit equivalence between $(\mu^{2})$ and uniform $\MCT$.  
\begin{thm}\label{sef}
From $\P\vdash \MCT_{\ns}\asa \paai$, terms $s, u$ can be extracted such that $\textup{\textsf{E-PA}}^{\omega*}$ proves:
\be\label{frood}
(\forall \mu^{2})\big[\textsf{\MU}(\mu)\di \MCT_{\ef}(s(\mu)) \big] \wedge (\forall t^{1\di 1})\big[ \MCT_{\ef}(t)\di  \MU(u(t))  \big].
\ee
\end{thm}
\begin{proof}
Apply $\CI$ to $\MCT_{\ns}\asa\paai$; The proof of the latter in \cite{samzoo}*{\S4.1} is rather elementary.  
\end{proof}

\subsubsection{Theorems equivalent to $\ATR_{0}$ and $\FIVE$}
In this section, we study equivalences relating to $\ATR_{0}$ and $\FIVE$, the strongest Big Five systems from RM.  
The associated results show that the template $\CI$ also works for the fourth and fifth Big Five system.  

\medskip

We shall work with the \emph{Suslin functional} $(S^{2})$, the functional version of $\FIVE$.  
\be\label{suske}
(\exists S^{2})(\forall f^{1})\big[   S(f)=_{0} 0 \asa (\exists g^{1})(\forall x^{0}) (f(\overline{g}x)\ne 0)\big]. \tag{$S^{2}$}
\ee
Feferman has introduced the following version of the Suslin functional (See e.g.\ \cite{avi2}).
\be\label{suske2}
(\exists \mu_{1}^{1\di 1})\big[(\forall f^{1})\big(    (\exists g^{1})(\forall x^{0}) (f(\overline{g}x)\ne 0)\di (\forall x^{0}) (f(\overline{\mu_{1}(f)}x)\ne 0)\big)\big], \tag{$\mu_{1}$}
\ee
where the formula in square brackets is denoted $\MUO(\mu_{1})$.  We shall require another instance of \emph{Transfer}:
\be\tag{$\Paai$}
(\forall f^{1})\big[    (\exists g^{1})(\forall x^{0}) (f(\overline{g}x)\ne 0)\di  (\exists^{\st} g^{1})(\forall^{\st} x^{0}) (f(\overline{g}x)\ne 0)\big].
\ee
%
%
We shall obtain an effective version of the equivalence proved in \cite{yamayamaharehare}*{Theorem~4.4}.  
The relevant (non-uniform) principle pertaining to the latter is $\PST$, i.e.\ the statement that \emph{every tree with uncountably many paths has a non-empty perfect subtree}.   
The latter has the following nonstandard and effective versions. 
\begin{thm}[$\PST_{\ns}$]For all standard trees $T^{1}$, there is standard $P^{1}$ such that 
\[
(\forall f_{(\cdot)}^{0\di 1})(\exists f\in T)(\forall n)(f_{n}\ne_{1} f) \di \textup{$P$ is a non-empty perfect subtree of $T$}.
\]
\end{thm}
\begin{thm}[$\PST_{\ef}(t)$]
For all trees $T^{1}$, we have
\[
(\forall f_{(\cdot)}^{0\di 1})(\exists f\in T)(\forall n)(f_{n}\ne_{1} f) \di \textup{$t(T)$ is a non-empty perfect subtree of $T$}.
\]
\end{thm}
As a technicality, we require that $P$ as in the previous two principles consists of a pair $(P', p')$ such that $P'$ is a perfect subtree of $T$ such that $p'\in P'$.      
We have the following theorem.
\begin{thm}\label{sef2334}
From $\P\vdash \PST_{\ns}\asa \Paai$, terms $s, u$ can be extracted such that $\textup{\textsf{E-PA}}^{\omega*}$ proves:
\be\label{frood3}
(\forall \mu_{1})\big[\textsf{\MUO}(\mu_{1})\di \PST_{\ef}(s(\mu_{1})) \big] \wedge (\forall t^{1\di 1})\big[ \PST_{\ef}(t)\di  \MUO(u(t))  \big].
\ee
\end{thm}
In light of the intimate connection between theorems concerning perfect kernels of trees and the Cantor-Bendixson theorem for Baire space (See \cite{simpson2}*{IV.1}), a version of Theorem \ref{sef2334} for the former can be obtained in a straightforward way.  
Another \emph{more mathematical} statement which can be treated along the same lines is \emph{every countable Abelian group is a direct sum of a divisible and a reduced group}.  The latter is called $\DIV$ and equivalent to $\FIVE$ by \cite{simpson2}*{VI.4.1}.  By the proof of the latter, the reverse implication is straightforward;  We shall study $\DIV\di \FIVE$.  

\medskip

To this end, let $\DIV(G, D, E)$ be the statement that the countable Abelian group $G$ satisfies $G=D\oplus E $, where $D$ is a divisible group and $E$ a reduced group.  
The nonstandard version of $\DIV$ is as follows:
\be\tag{$\DIV_{\ns}$}
(\forall^{\st}G)(\exists^{\st} D, d, E)\big[\DIV(G, D, E)\wedge (D\ne\{0_{G}\}\di d\in D)\big],
\ee
where we used the same technicality as for $\PST_{\ns}$.  The effective version is:
\be\tag{$\DIV_{\ef}(t)$}
(\forall G)\big[\DIV(G, t(G)(1), t(G)(2))\wedge (t(G)(1)\ne\{0_{G}\}\di t(G)(3)\in t(G)(1))\big].
\ee
We have the following (immediate) corollary.  
\begin{cor}\label{sef23345}
From $\P\vdash\DIV_{\ns}\di \Paai$, a term $u$ can be extracted such that $\textup{\textsf{E-PA}}^{\omega*}$ proves:
\be\label{frood3555}
(\forall t^{1\di 1})\big[ \DIV_{\ef}(t)\di  \MUO(u(t))  \big].
\ee
\end{cor}

\subsection{The Reverse Mathematics zoo}
The Reverse Mathematics zoo is a collection of theorems which do not fit the `Big Five' categories (\cite{damirzoo}).   
In \cite{samzoo, samzooII}, a variant of $\CI$ is used to classify \emph{uniform} versions of the RM zoo as equivalent to arithmetical comprehension $(\exists^{2})$.  We list the relevant results 
for one theorem from the RM zoo, namely $\DNR$ as defined below.  All known theorems from the RM zoo have been classified in the same way.  

\medskip

Thus, consider the principle \textsf{UDNR} as follows: $(\exists \Psi^{1\di1})\big[(\forall A^{1})(\forall e^{0})(\Psi(A)(e)\ne \Phi_{e}^{A}(e))\big]$.
Clearly, $\UDNR$ is the uniform version of the zoo principle\footnote{We sometimes refer to inhabitants of the RM zoo as `theorems' and sometimes as `principles'.} $\DNR$ defined as:  
$(\forall A^{1})(\exists f^{1})(\forall e^{0})\big[f(e)\ne \Phi_{e}^{A}(e)\big].$
The principle $\DNR$ was first formulated in \cite{withgusto} and is even strictly implied by $\textsf{WWKL}$ (See \cite{compdnr}) where the latter principle sports \emph{some} Reverse Mathematics equivalences (\cites{montahue, yuppie, yussie}) but is not a Big Five system.  Nonetheless, it is the case that $\UDNR\asa (\exists^{2})$. 
In other words, \emph{the `exceptional' status of $\DNR$ disappears completely if we consider its uniform version $\UDNR$}.  

\medskip

To prove that $\UDNR$ is equivalent to arithmetical comprehension, we consider $\UDNR^{+}$:
\[
(\exists^{\st}\Psi^{1\di 1})\big[(\forall^{\st} A^{1})(\forall e^{0})(\Psi(A)(e)\ne \Phi_{e}^{A}(e))\wedge (\forall^{\st} C^{1}, D^{1})\big(C\approx_{1} D \di \Psi(C)\approx_{1}\Psi(D) \big)\big], 
\]
where $A\approx_{1}B$ if $(\forall^{\st}n)(A(n)=B(n))$.  
The second conjunct expresses that $\Psi$ is \emph{standard extensional}.  
\begin{thm}\label{UDNRPLUS}
In $\P$, we have $\UDNR^{+}\asa \paai$.  
\end{thm}
Denote by $\UDNR(\Psi)$ the formula in square brackets in \textsf{UDNR}.  
\begin{thm}\label{sef2}
From $\P\vdash\UDNR^{+}\asa \paai$ terms $s, u$ can be extracted such that $\textsf{\textup{E-PA}}^{\omega*}$ proves:
\be\label{frood9}
(\forall \mu^{2})\big[\textsf{\MU}(\mu)\di \UDNR(s(\mu)) \big] \wedge (\forall \Psi^{1\di 1})\big[ \UDNR(\Psi)\di  \MU(u(\Psi, \Xi))  \big],
\ee
where $\Xi$ satisfies  $(\forall A^{1}, B^{1}, k^{0})(\overline{A}\Xi(A, b, k)=\overline{B}\Xi(A, B, k)  \di \Psi(A)(k)=\Psi(B)(k))   $.
\end{thm}
In the previous theorem, we say that $\Xi$ \emph{is an extensionality functional} for $\Psi$, as the former witnesses the axiom of extensionality for the latter.  
Proofs of the previous theorems may be found in \cite{samzoo}, while a general template to similarly treat theorems from the Reverse Mathematics zoo may be found in \cite{samzoo, samzooII}.
As it turns out, these proofs also go through relative to Heyting arithmetic (\cite{samzooII}).

\section{The Gandy-Hyland functional}\label{GH2}
In this section, we apply $\CI$ to computability theory by studying the \emph{Gandy-Hyland functional}.  Proofs and additional results may be found in \cite{samGH}.   
\subsection{Introducing the Gandy-Hyland functional $\Gamma$}
The Gandy-Hyland functional was introduced in \cite{gandymahat} as an example of a higher-type functional 
not computable, in the sense of Kleene's S1-S9 (See \cite{noortje}*{1.10} or \cite{longmann}*{5.1.1}), in the \emph{fan functional} over the total continuous functionals (See \cite{noortje}*{4.61} or \cite{longmann}*{8.3.3}).   
The {Gandy-Hyland functional} $\Gamma$ is:  
\begin{equation}\label{GH}\tag{\textsf{GH}}
(\exists \Gamma^{3})(\forall Y^{2}\in C, s^{0})\big[\Gamma(Y^{2},s^{0})= Y\big(s*0* (\lambda n^{0})\Gamma(Y, s*(n+1))\big)\big], 
\end{equation}
where `$Y^{2}\in C$' is the usual definition of pointwise continuity on Baire space as in \eqref{CB2}.   
We adopt the usual notational conventions as in e.g.\ \cite{bergolijf}.
\be\label{CB2}
(\forall f^{1})(\exists N^{0})(\forall g^{1})(\overline{f}N=_{0}\overline{g}N\di Y(f)=_{0}Y(g)).  
\ee
The functional $\Gamma$ from \eqref{GH} apparently exhibits non-well-founded self-reference:  Indeed, in order to compute $\Gamma$ at $s^{0}$, one needs the values of $\Gamma$ at all child nodes of $s^{0}$, as is clear from the right-hand side of \eqref{GH}.  
In turn, to compute the value of $\Gamma$ at the child nodes of $s$, one needs the value of $\Gamma$ at all grand-child nodes of $s$, and so on.  Hence, repeatedly applying the definition of $\Gamma$  seems to result in a non-terminating recursion.       
By contrast, \emph{primitive recursion} is well-founded as it reduces the case for $n+1$ to the case for $n$, and the case for $n=0$ is given.  

\medskip

As it turns out, the Gandy-Hyland functional as in \eqref{GH} can be approximated in Nonstandard Analysis by the following \emph{primitive recursive}\footnote{The functional $G$ is primitive recursive \emph{in the sense of G\"odel's system} $\textsf{T}$ by \cite{escaleert}*{Theorem 18}.} functional:
\be\label{small}
G(Y,s,M)=
\begin{cases}
Y(s*00\dots) & |s|\geq M \\
Y(s*0* (\lambda n^{0})G(Y, s*(n+1),M)) & \textup{otherwise}   
\end{cases}
\ee
Indeed, $G$ as in \eqref{small} equals the $\Gamma$-functional from \eqref{GH} for standard input and any \emph{nonstandard number} $M^{0}$ (See Section \ref{korium}).  
Note that one need only apply the definition of $G$ at most $M$ times to terminate in the first case of \eqref{small}.  In other words, the extra case `$|s|\geq M$' 
provides a nonstandard stopping condition which `unwinds' the non-terminating recursion in $\Gamma$ to the terminating one in $G$.  Or: one can trade in self-reference for nonstandard numbers.  
Thus, we shall refer to $G$ as the \emph{canonical approximation} of $\Gamma$.    

\medskip

To be absolutely clear, all systems mentioned in this paper deal with \emph{total functionals only}.  In particular, in the system $\P+\eqref{GH}$, there is a functional $\Gamma^{3}$ which behaves as 
described in \eqref{GH} for $Y^{2}\in C$, while $\Gamma(Z, s)$ is a natural number for discontinuous $Z^{2}$ and $s^{0}$, but we have no additional information.  The same convention applies to the 
modulus-of-continuity functional defined in the next section.  

\subsection{Term extraction and $\Gamma$}\label{korium}
In this section, we show that $\Gamma$ equals its canonical approximation $G$ {assuming certain fragments of \emph{Transfer} and \emph{Standard Part}} as in Theorem \ref{Cruxxxx}.  Applying $\CI$ to this result, one obtains a term $t$ expressing the Gandy-Hyland functional in terms of the \emph{modulus-of-continuity functional} and a special case of the fan functional, as in Corollary \ref{TE}.  
To this end, we need the following nonstandard axioms.  
\be\tag{$\GH_{\st}(\Gamma)$}
(\forall^{\st} Y^{2}\in C, s^{0})\big[\Gamma(Y^{2},s^{0})= Y\big(s*0* (\lambda n^{0})\Gamma(Y, s*(n+1))\big)\big].
\ee
\be\tag{\textup{$\sigtoe$}}\label{predruk}
(\forall^{\st}f^{1})\big[(\exists m^{0})(\forall n^{0})f(m,n)=0\di (\exists^{\st} k^{0})(\forall l^{0})f(k,l)=0\big].
\ee
\be\label{STP}\tag{\textup{\textsf{STP}}}
(\forall f^{1}\leq_{1}1)(\exists^{\st} g^{1}\leq_{1}1)(f\approx_{1}g).
\ee   
The notation $f\approx_{1}g$ stands for `$(\forall^{\st}n^{0})(f(n)=_{0}g(n))$', and the the function $g^{1}$ from \ref{STP} is called a \emph{standard part} of $f^{1}$.  Clearly, \eqref{STP} is a fragment of \emph{Standard Part}, while $\sigtoe$ is a fragment of \emph{Transfer}.  Proofs of the following theorem and corollary may be found in \cite{samGH}*{\S4.1}.    
\begin{thm}\label{Cruxxxx}
In $\P+\sigtoe+\STP$, the Gandy-Hyland functional exists and equals its canonical approximation, i.e.\ there is standard $\Gamma^{3}$ such that $\GH_{\st}(\Gamma)$ and
\be\label{crubm}\tag{\textsf{\textup{CA}}$(\Gamma)$}
(\forall^{\st}Y^{2}\in C,s^{0})(\forall N\in \Omega)(G(Y,s,N)=\Gamma(Y,s)).
\ee
Furthermore, the Gandy-Hyland functional is unique, i.e.\ $(\forall \Gamma_{1}^{3})(\GH_{\st}(\Gamma_{1})\di \textsf{\textup{CA}}(\Gamma_{1}))$.
\end{thm}
To apply term extraction to Theorem \ref{Cruxxxx}, the following principles are needed.   
\be
(\forall Y^{2}\in C, f^{1}, g^{1})(\overline{f}\Psi(Y, f)=\overline{g}\Psi(Y, f)\notag \di Y(f)=Y(g)). \label{lukl2}\tag{$\textsf{\textup{MPC}}(\Psi)$}
\ee
Note that $\MPC(\Psi)$ states that $\Psi^{3}$ is a modulus-of-continuity functional, while $\MU(\mu)$ states that $\mu^{2}$ is Feferman's search operator (See e.g.\ \cite{kohlenbach2} for the latter).  
We also need the following functional.    
\begin{align}\label{fanns333}\tag{$\SCF(\Theta)$}
(\forall g^{2}, T^{1}\leq_{1}1)\big[ (\forall  \alpha^{1}\in \Theta(g)(2))(\alpha\leq_{1}1\di& \overline{\alpha}g(\alpha)\not\in T)\di\\
 &(\forall \beta\leq_{1}1)(\exists i\leq_{0}\Theta(g)(1))(\overline{\beta}i\not\in T) \big].\notag
\end{align}
The functional $\Theta^{3}$ as in $\SCF(\Theta)$ is called \emph{the special fan functional}, and its properties are discussed in Section \ref{norke}.  
For now, it suffices to know that the special functional is part of classical and Brouwerian intuitionistic mathematics.  
Note that there is no unique $\Theta$ as in $\SCF(\Theta)$, i.e.\ it is in principle incorrect to talk about `the' special fan functional.    
Finally, let $\GH(\Gamma)$ be \eqref{GH} with the leading quantifier omitted.  
\begin{cor}[Term Extraction]\label{TE}
From the proof in $\P$ of 
\be\label{nsaversion}
\ref{predruk}+\ref{STP}\di (\forall \Gamma^{3})\big[\textsf{\textup{GH}}_{\st}(\Gamma)\di \textup{\textsf{CA}}(\Gamma)   \big], 
\ee
a term $t^{4}$ can be extracted such that \textsf{\textup{E-PA}}$^{\omega*}+\QFAC^{1,0}$ proves that
\begin{align}\label{consversion}
(\forall \mu^{2} , \Theta^{3},\Gamma^{3} )\big[\big(\textsf{\textup{GH}}(\Gamma)  \wedge &\textup{\textsf{MU}}(\mu)\wedge \SCF(\Theta)\big) \di (\forall Y^{2} \in C, s^{0})\big(G(Y, s, t(Y,s, \mu,\Theta))=\Gamma(Y, s) \big)\big], \notag
\end{align}
i.e.\ $G(Y, s, t(Y, s, \mu,\Theta ))$ is the Gandy-Hyland functional expressed in terms of Feferman's search operator and the special fan functional.  
\end{cor}
\begin{proof}
Apply $\CI$ to the proof in Theorem \ref{Cruxxxx}.  
\end{proof}
Note that Feferman's search operator can be defined in terms of a modulus-of-continuity functional (and vice versa) by combining the results in \cite{kohlenbach2}*{\S3}, \cite{exu}, and \cite{kohlenbach4}*{\S4}.  Hence, we have the following:  
\begin{cor}
In the system from the previous corollary, the Gandy-Hyland functional can be expressed in terms of a modulus-of-continuity functional and the special fan functional.  
\end{cor}

\subsection{Term extraction and $\Gamma$, again}
In this section, we show that the results of the previous section are rather \emph{modular}, in that we may obtain variations of Theorem \ref{Cruxxxx} and its corollaries.   
To this end, let $\textsf{NPC}(Y)$ be the following formula:
\be\tag{$\textup{\textsf{NPC}}(Y)$}
(\forall^{\st}f^{1})(\forall g^{1})(f\approx_{1}g \di Y(f)=_{0}Y(g)), 
\ee
i.e.\ the previous formula expresses that $Y^{2}\in C$ is nonstandard continuous.  Furthermore, let $\ST(\Gamma, Y)$ be $(\forall^{\st}s^{0})(\st(\Gamma(Y,s)))$, i.e.\ $\Gamma$ produces standard outputs, and let 
$\GH(\Gamma, Y)$ be \eqref{GH} with the two leading quantifiers omitted.  
\begin{thm}\label{ragnor}
The system $\P+\STP$ proves that for all $\Gamma^{3}$ and $Y^{2}$
\be\label{drfre}
\big[\NSC(Y) \wedge \GH(\Gamma, Y) \wedge \textsf{\textup{ST}}(\Gamma, Y)\big]\di (\forall^{\st}s^{0})(\forall N\in \Omega)(\Gamma(Y,s)=G(Y, s, N)). 
\ee
\end{thm}
We need the following principles;  Note that $\PCM(Y^{2},Z^{2})$ expresses that $Z$ is a modulus of pointwise continuity for $Y$.
\be\tag{$\PCM(Y, Z)$}
(\forall f^{1}, g^{1})(\overline{f}Z(f)=_{0}\overline{g}Z(f)\di Y(f)=_{0}Y(g))
\ee
\be\tag{$\GHU(\Gamma, Y, H)$}
(\forall  s^{0})\big[\Gamma(Y, s)=Y(s*0*(\lambda n)\Gamma(Y, s*(n+1)))\wedge \Gamma(Y, s)\leq H(Y, s)\big],
\ee

Applying $\CI$ to the previous theorem, one obtains the following theorem.
\begin{cor}[Term Extraction]\label{boei}
From the proof in Theorem \ref{ragnor}, a term $t$ can be extracted such that $\textsf{\textup{E-PA}}^{\omega*}+\QFAC^{1,0}$ proves for $\Xi=(H^{1},Z^{2},\Theta^{3})$ and $\Gamma^{3}, Y^{2}$ that
\[
\big[\PCM(Y,Z)\wedge \SCF(\Theta) \wedge \GHU(\Gamma, Y,H) \big]\di (\forall s)(\forall N\geq t(s,\Xi))(\Gamma(Y,s)=G(Y, s, N)),  
\]
i.e.\ the Gandy-Hyland functional $\Gamma$ at $Y$ can be approximated via a modulus of continuity of $Y$, the special fan functional, and an upper bound for $\Gamma(Y,\cdot)$.  
\end{cor}
\subsection{The special fan functional}\label{norke}
We discuss some surprising (computational and otherwise) properties of the special fan functional, which was first introduced in \cite{samGH}*{\S3}.  

\medskip

First of all, the full axiom \emph{Transfer} of $\IST$ does not imply the full axiom \emph{Standard Part} (over various systems; see \cites{blaaskeswijsmaken, gordon2}).
In this light, it is a natural question whether the same holds for prominent fragments discussed in this paper.  For instance, are $\paai\di \STP$ or $\Paai\di \STP$ provable in $\P$?

\medskip

Secondly, the previous questions can be translated into relative computability questions regarding the special fan functional and functionals like $(\mu^{2})$.  
We briefly discuss the answers to these questions from \cite{norsa}.     
We need the following functionals.  
\be
(\forall Y^{2}) (\forall f^{1}, g^{1}\leq_{1}1)(\overline{f}\Phi(Y)=\overline{g}\Phi(Y)\notag \di Y(f)=Y(g)). \label{lukl3}\tag{$\textsf{\textup{MUC}}(\Phi)$}
\ee
\be\tag{$\mathcal{E}_{2}$}\label{hah}
(\exists \xi^{3})(\forall Y^{2})\big[  (\exists f^{1})(Y(f)=0)\asa \xi(Y)=0  \big].
\ee
The functional $\Phi^{3}$ as in $\MUC(\Phi)$ is called the \emph{intuitionistic} \emph{fan functional} and yields a conservative extension of weak K\"onig's lemma for the second-order language (See \cite{kohlenbach2}*{Prop.\ 3.15}).  By the following theorems, the special fan functional is an object of intuitionistic and classical mathematics.
\begin{thm}
There is a term $t$ such that $\textsf{\textup{E-PA}}^{\omega}$ proves $(\forall \Omega^{3})(\MUC(\Omega)\di \SCF(t(\Omega)))$.  
\end{thm}
\begin{thm}[$\ZFC$]\label{nor1}
A functional $\Theta^{3}$ as in $\SCF(\Theta)$ can be computed (Kleene's S1-S9) from $\xi$ as in $(\mathcal{E}_{2})$. 
\end{thm}
\begin{thm}[$\ZFC$]\label{nor2}
Let $\varphi^{2}$ be any type two functional.   
Any functional $\Theta^{3}$ as in $\SCF(\Theta)$ is not computable (Kleene S1-S9) in $\varphi^{2}$.
\end{thm}
Theorems \ref{nor1} and \ref{nor2} were first proved by Normann and are forthcoming in \cite{norsa}.  Theorem \ref{nor2} for the special case of $(\mu^{2})$ originates from the conjecture by the author that $\paai$ does not imply $\STP$ over $\P$.    Since the Suslin functional is of type two, it cannot compute (Kleene S1-S9) the special fan functional, which translates back to the fact that $\P$ does not prove $\Paai\di \STP$.

\begin{ack}\rm
This research was supported by the following funding bodies: FWO Flanders, the John Templeton Foundation, the Alexander von Humboldt Foundation, the University of Oslo, and the Japan Society for the Promotion of Science.  
The author expresses his gratitude towards these institutions. 
The author would like to thank Grigori Mints, Ulrich Kohlenbach, Horst Osswald, Helmut Schwichtenberg, Stephan Hartmann, Dag Normann, and Karel Hrbacek for their valuable advice.  
Finally, I thank the anonymous referees for their many helpful suggestions.  
\end{ack}

%
%
%

\section{Bibliography}
\bibliographystyle{eptcs}


\end{document}